\documentclass[11pt,a4paper]{article}
\usepackage{fullpage}
\usepackage{amsmath, amsthm, amsfonts,  amssymb, latexsym}
\usepackage[colorlinks, citecolor=Red, linkcolor=Blue, urlcolor=blue]{hyperref}
\usepackage{caption}
\usepackage[dvipsnames]{xcolor}
\usepackage{tikz}
\usepackage{enumerate}
\usepackage{cancel}

\usepackage{authblk}

\newtheorem{proposition}{Proposition}
\newtheorem{theorem}[proposition]{Theorem}
\newtheorem{lemma}[proposition]{Lemma}

\newtheorem{corollary}[proposition]{Corollary}
\newtheorem{remark}[proposition]{Remark}

\newcommand{\rhs}{r.h.s.\ }
\newcommand{\lhs}{l.h.s.\ }
\newcommand{\wrt}{w.r.t.\ }
\newcommand{\cf}{cf.\ }

\newcommand{\ibar}{{\frac{i}{\hbar}}}

\newcommand{\eti}[1]{e_\otimes^{\ibar #1}}
\newcommand{\et}[1]{e_\otimes^{#1}}

\newcommand{\ud}{\mathrm{d}}

\newcommand{\g}{\mathfrak{g}}
\newcommand{\p}{\mathfrak{p}}
\newcommand{\skal}[2]{\langle #1 , #2 \rangle}

\newcommand{\LB}[2]{\lfloor #1 , #2 \rfloor}

\newcommand{\sS}{\mathcal{S}}

\newcommand{\cR}{\mathcal{R}}

\newcommand{\cD}{\mathcal{D}}

\newcommand{\eps}{\varepsilon}

\newcommand{\N}{\mathbb{N}}

\newcommand{\ia}{{\mathrm{int}}}
\newcommand{\ret}{{\mathrm{r}}}

\newcommand{\lin}{{\mathrm{lin}}}

\newcommand{\vp}{{\varphi}}

\newcommand{\nn}{\nonumber}
\newcommand{\beq}{\begin{equation}}
\newcommand{\eeq}{\end{equation}}

\newcommand{\defeq}{\mathrel{:=}}

\newcommand{\vol}{\mathrm{vol}}

\DeclareMathOperator{\Tr}{Tr}
\DeclareMathOperator{\supp}{supp}
\DeclareMathOperator{\Deg}{Deg}

\newcommand{\cO}{{\mathcal{O}}}
\newcommand{\cA}{{\mathcal{A}}}

\newcommand{\ba}{{\bar a}}

\begin{document}

\title{Background independence and the Adler-Bardeen theorem}
\author{Jochen Zahn \\ Institut f\"ur Theoretische Physik, Universit\"at Leipzig\\ Br\"uderstr.\ 16, 04103 Leipzig, Germany \\ jochen.zahn@itp.uni-leipzig.de}

\date{\today}

\maketitle

\begin{abstract}
We prove that for renormalizable Yang-Mills gauge theory with arbitrary compact gauge group (of at most a single abelian factor) and matter coupling, the absence of gauge anomalies can be established at the one-loop level.
This proceeds by relating the gauge anomaly to perturbative agreement, which formalizes background independence.
\end{abstract}

\section{Introduction}

In its most general form, the Adler-Bardeen theorem states that in perturbative, renormalizable Yang-Mills gauge theory, gauge anomalies are absent if they are absent at the one-loop, or $\cO(\hbar)$, level.\footnote{The original statement \cite{AdlerBardeen}, see also \cite{Zee72, LowensteinSchroer73, Kopper:2011aa}, did actually not concern the gauge anomaly, but the chiral anomaly in QED. Recently, also a non-perturbative version of this was proven for lattice QED \cite{Mastropietro2020}.} It has been argued for at various levels of rigour and generality \cite{CostaEtAl77, BandelloniEtAl80, PiguetSorella92, PiguetSorella93, AnselmiAdlerBardeen}.\footnote{Recently, also non-perturbative (lattice) results in this direction were established \cite{Mastropietro:2020bhz, Mastropietro:2021dxm, Mastropietro:2023wnb}. However, not the full field content of the Standard Model is included and only a $U(1)$ gauge symmetry is considered in these approaches. Hence, while having a broader scope (in the sense of being non-perturbative), the mentioned results apply to a more restricted class of theories, not (yet) covering the Standard Model, contrary to our approach.} With the exception of \cite{AnselmiAdlerBardeen}, the Callan-Symanzik equation \cite{Callan70, Symanzik70}, and the quantum action principle \cite{Lowenstein71, Lam72} are employed.\footnote{By contrast, in \cite{AnselmiAdlerBardeen} a particular, rather non-trivial, regularization scheme is employed in which the anomaly cancellation to all orders is argued to be manifest.} However, the arguments put forward are not completely satisfactory: They are based on a particular (vacuum) state, while in generic backgrounds (of gravitational and/or gauge type), there is no preferred state. 
Furthermore, the generating (vertex) functionals used in these approaches are formal, as they are integrated over the entire spacetime, leading to potential infrared issues.
Finally, the ``algebraic'' approaches \cite{PiguetSorella92, PiguetSorella93} seem to apply only to the non-abelian anomaly, but not to abelian anomalies (which, for example, are most crucial to establish that anomaly freedom essentially fixes the matter representation in the Standard Model \cite{GengMarshak}).\footnote{Other approaches  \cite{CostaEtAl77, BandelloniEtAl80, AnselmiAdlerBardeen} also apply to some of the possible abelian anomalies (the first two of those listed in \eqref{eq:AbelianChiral} below), but apparently not to general abelian anomalies, see also Remark~\ref{rem:AbelianAnomaly} below.}

It is thus desirable to prove the Adler-Bardeen theorem in the rigorous framework of locally covariant field theory \cite{HollandsWaldWick, HollandsWaldTO, HollandsYM}, \cf also \cite{HollandsWaldReview, RejznerBook} for recent reviews or monographs. The author's attempts to adapt the above mentioned approaches by replacing the Callan-Symanzik equation by renormalization group transformations as defined in \cite{HollandsWaldRG} failed. However, it turns out that there is a close connection between gauge anomalies and anomalies of background independence, which one can exploit to prove the Adler-Bardeen theorem. Background independence is formalised in \emph{perturbative agreement} \cite{HollandsWaldStress, BackgroundIndependence}, which essentially states that quadratic terms can be equivalently included in the free or the interacting part of the action, see also \cite{BgInd} for a discussion. Typically, one thinks of these quadratic terms as arising from the free action by variations of some of the background geometric data, such as the metric \cite{HollandsWaldStress} or a background gauge connection \cite{BackgroundIndependence}. For our purposes, it is the latter notion which is of relevance.

A first indication of a connection between gauge anomalies and perturbative agreement is the fact that the latter actually has the Adler-Bardeen property: Perturbative agreement can be fulfilled to all orders, if a certain obstruction, visible at $\cO(\hbar)$, vanishes \cite{HollandsWaldStress, BackgroundIndependence}. Furthermore, the obstruction is, for the case of variations of a background gauge connection, the covariant divergence $T( \bar \nabla_\mu J_0^{\mu I})$ of the corresponding renormalized free current. 
Here $\bar \nabla$ indicates the covariant derivative \wrt the background gauge connection,  $J_0^{\mu I} \defeq \frac{\delta}{\delta \bar \cA^I_\mu} S_0$ is the functional derivative of the free part of the action \wrt the background gauge connection $\bar \cA$ (with $I$ a Lie algebra index), and $T$ denotes a renormalization scheme (a time-ordered product in the terminology used below).\footnote{Note that a time-ordered product with a single factor (as considered above) is often called a (renormalized) Wick product. However, in non-trivial backgrounds, already the renormalization of such Wick powers is non-trivial and ambiguous \cite{HollandsWaldWick} and it is very convenient to denote renormalized Wick powers and time-ordered products by the same symbol (for example in order to express certain relations between them, such as \eqref{eq:CausalFactorization} below, in a uniform notation).} In more formal approaches a non-vanishing $T( \bar \nabla_\mu J_0^{\mu I})$ is called a \emph{covariant anomaly}, and it has (formally) a well-understood relation \cite{BardeenZumino} to the \emph{consistent anomaly}, which governs the gauge anomaly. Hence, in practice, when determining violations of perturbative agreement, one computes the well-known anomalies, just in covariant form. 

The fact that an Adler-Bardeen type theorem has already been proven for perturbative agreement suggests the following strategy for proving the Adler-Bardeen theorem: When we assume vanishing of the covariant divergence $T( \bar \nabla_\mu J_0^{\mu I})$ of the renormalized free current (which is an $\cO(\hbar)$ statement, as $J_0$ is quadratic in the fields so that application of $T$ yields a further c-number term of $\cO(\hbar)$), we may also assume fulfillment of perturbative agreement \wrt changes in the background connection. These two properties turn out to be sufficient to prove absence of anomalies. To summarize, we prove the following:

\begin{theorem}
\label{thm:main}
Assume that, in a perturbatively renormalizable Yang-Mills gauge theory with a compact gauge group of at most a single $U(1)$ factor and arbitrary matter coupling on a globally hyperbolic four-dimensional spacetime, there exists a locally covariant renormalization scheme such that the covariant divergence $T( \bar \nabla_\mu J_0^{\mu I})$ of the free current vanishes (on-shell, when anti-fields are set to zero) for generic on-shell background gauge fields and vanishing background matter fields. Then there exists a locally covariant renormalization scheme in which gauge anomalies are absent for generic on-shell background gauge fields and vanishing background matter fields.
\end{theorem}

The terms ``locally covariant renormalization scheme'' and ``gauge anomaly'' will be explained below (the former one at least as far as necessary for our purposes). Let us first remark on the assumptions and the implications of the theorem:

\begin{itemize}

\item The restriction to at most a single $U(1)$ factor in the gauge group ensures the absence of so-called exceptional anomalies \cite{BarnichBrandtHenneaux00}. This assumption can be relaxed when the occurrence of the exceptional anomalies can be excluded otherwise. In any case, as the gauge group of the Standard Model only contains a single $U(1)$ factor, the theorem is applicable to the Standard Model (at least in the unbroken phase, see next item).

\item The restriction to  trivial matter background implies no restrictions on the matter action (except for the absence of a linear ``source'' term for the matter fields in the action), as the discussion is limited to the algebraic level, i.e., existence of states is not required. In particular, the vanishing matter background need not be stable (as long as it is a solution to the classical field equations). Nevertheless, in view of the Higgs mechanism in the Standard Model, an extension to non-trivial bosonic background matter fields is desirable, and will be sketched after establishing our main result.

\item The existence of a locally covariant renormalization scheme on globally hyperbolic spacetimes is well established \cite{HollandsWaldTO, HollandsWaldReview} (see also \cite{LocCovDirac} for the generalization to background gauge connections), irrespective of the gauge group or the matter content. Hence, the non-trivial crucial assumption is the vanishing of the covariant divergence $T( \bar \nabla_\mu J_0^{\mu I})$ of the free current. As discussed above, if it is satisfied, we may also assume perturbative agreement, which in turn will be used to rule out the occurrence of gauge anomalies. The crucial point is that this condition is easy to check in concrete examples, as it can be reduced to studying violations of the Hadamard parametrix being a solution to the equation of motion, as in \cite{ChiralFermions} for the case of chiral fermions in background gauge fields. Concretely, we can split the free current $J_0^{\mu I}$ into a contribution from the matter Lagrangian, and a contribution which only contains gauge and the corresponding auxiliary fields. That the covariant divergence of the latter part vanishes on-shell can be shown without detailed analysis of the Hadamard parametrix.\footnote{The ``gauge part'' of the free current has no component in the abelian direction, and this property is conserved by Hadamard point-split renormalization, so we can restrict consideration to the semi-simple part of the Lie algebra. The covariant divergence of this part of the current is a Lie algebra valued function (or rather a section of the bundle $\p$ defined below) of mass dimension four. 
However, for a semi-simple Lie algebra, one can not construct such sections covariantly out of the background fields, using the structure constants of the adjoint representation: Only the background field strength $\bar F^I_{\mu \nu}$ carries a Lie algebra index, and to contract it to a Lorentz scalar, one needs either $\bar F^I_{\mu \nu}$ or $* \bar F^I_{\mu \nu}$ (the star indicating the Hodge dual). But as the symmetric structure constants of the adjoint representation vanish, one can not construct a possible violation of current conservation in this way.}
The covariant divergence of the matter part for (chiral) fermions was calculated in \cite[Sect.~5]{ChiralFermions} 
and it vanishes if and only if the well-known criteria for the cancellation of the ``triangle anomalies'' are fulfilled, namely $\Tr [ \gamma^5 T^I \{ T^J, T^K \} ] = 0$ and\footnote{The latter condition is necessary for the cancellation of mixed gauge/gravitational anomalies.} $\Tr[ \gamma^5 T^I ] = 0$ with $T^I$ the generators of the matter representation and the trace over spinor and representation indices. That no contributions to the divergence of the free current arise from scalar matter fields can be easily seen by adapting the results of \cite{ChiralFermions}. Hence, we see that our result ensures the absence of gauge anomalies in the Standard Model.

\item ``Absence of gauge anomalies'' is to be understood in the usual sense of perturbative gauge theory: The gauge anomaly can be removed (inductively, order by order in $\hbar$) to arbitrarily high order in $\hbar$.

\end{itemize}

The proof of the above theorem is based on the following strategy: As is well-known \cite{BarnichBrandtHenneaux00, HollandsYM}, but also recalled below, the potential gauge anomalies can be classified cohomologically as elements of $H^1_4(s | \ud)$, i.e., 4-forms of ghost number 1, which are closed under the BV differential $s$, up to $\ud$ exact terms. By our assumption on the gauge group a potential anomaly is a linear combination of an abelian anomaly (a product of the abelian ghost and a gauge invariant field) and a non-abelian anomaly \cite{BarnichBrandtHenneaux00}. The general strategy is to show that both these types of anomalies are inconsistent with perturbative agreement. The first step of the proof is to show that a trivial gauge anomaly, i.e., a trivial element of $H^1_4(s | \ud)$, can be removed by redefinitions of time-ordered products (this is well known), while preserving perturbative agreement (this needs to be established). In the next step, we rule out abelian anomalies. Using the fact that the abelian ghost enters the interaction only through a ``source term'' generating abelian gauge transformations of the matter fields, we can relate an abelian anomaly to the divergence of the interacting current (the interacting version of the free current discussed above). By perturbative agreement, this vanishes, implying the absence of abelian anomalies. In the last step, we rule out non-abelian anomalies. This proceeds by noting that, due to perturbative agreement and results of \cite{BgInd}, the anomaly $a$ is background independent in cohomology, i.e., $\hat \cD a = s b + \ud c$, where $\hat \cD$ is a differential which measures the background dependence of fields (it vanishes when acting on a field which is background independent, i.e., does not depend on how the gauge connection is split into a background and a dynamical part). The field $b$ in the above identity can be explicitly determined for the non-abelian anomaly. But, again by perturbative agreement, also $\hat \cD b$ must be trivial in cohomology. This, however, is not fulfilled for the concrete $b$ corresponding to the non-abelian anomaly. This rules out the non-abelian anomaly and concludes the proof.

In the following, we will first recall some results of \cite{HollandsYM, BgInd} on gauge anomalies and background independence and prove some elementary lemmata. We will then, in Section~\ref{sec:Proof}, prove Theorem~\ref{thm:main}. In the concluding Section~\ref{sec:Conclusion} we in particular sketch the generalization to non-trivial bosonic background matter fields.

\section{Gauge anomalies and background independence}

We are working on general four-dimensional globally hyperbolic spacetimes $(M, g)$, equipped with an orientation and a time-orientation.
We consider a compact Lie group $G$ with at most a single abelian factor, i.e., a quotient by a discrete subgroup of a product of simple groups and at most one $U(1)$, and denote its Lie algebra by $\g$. On the latter, we choose an invariant negative definite symmetric bilinear form, which on the simple factors is a multiple of the Killing form. We consider principal $G$ bundles $P$ with base manifolds $M$, and denote by $\p$ the associated $\g$ vector bundle (whose sections would be ``Lie-algebra valued functions'' in physics parlance). On $P$, we consider connections $\cA$ whose curvatures $F$ are sections of $\p \otimes \Omega^2$ (``Lie-algebra valued two-forms''). Given a representation $\rho$ of $G$, we consider matter fields $\psi$ as sections of the associated vector bundle.\footnote{When spin $\frac{1}{2}$ matter fields are considered, $M$ should also be equipped with a spin structure.} The starting point of our considerations is then an action of the form
\beq
\label{eq:Action}
 \tilde S = \frac{1}{4} \int F^I_{\mu \nu} F^{I \mu \nu} \vol + S_m(\psi, \nabla \psi),
\eeq
with $\vol$ the volume four-form, indices $I$ labelling an appropriately normalized basis of $\g$, and $\nabla$ the covariant derivative \wrt the connection $\cA$.

Given a background configuration $\bar \cA$ of the gauge connection $\cA$, we consider perturbations $A$, i.e., we write
\beq
 \cA = \bar \cA + A
\eeq
where $A$ is a section of\footnote{The difference of any two principal bundle connections $\cA$, $\cA'$ can be naturally interpreted as a section of this bundle.} $\p \otimes \Omega^1$ (a ``Lie algebra valued one-form''). We then consider $\bar \cA$ as a geometric datum on the same footing as the metric $g$ (see \cite{LocCovDirac} for the formulation of local covariance including background gauge connections). An important case is that of $\bar \cA$ being on-shell, i.e., a configuration extremizing \eqref{eq:Action} with $\psi = 0$. Expanding $\tilde S$ in $A$, the component $\tilde S_{-1}$ which is of first order in $A$ then drops out. We shall from now on restrict to that case and comment on deviations for off-shell backgrounds below. We denote by $\sS$ the ``manifold'' of on-shell backgrounds and by $T \sS$ its tangent space, consisting of infinitesimal perturbations $\bar a$ of $\bar \cA$, which are solutions to the linearized equations of motion.\footnote{Strictly speaking, the set of solutions to the Yang-Mills equation is a manifold only up to singular points corresponding to solutions possessing certain symmetries \cite{Arms1981}, i.e., the linearized equations around these special configurations have solutions which do not correspond to infinitesimal variations of the full solution. This does not affect our considerations, as we work locally in configuration space so that we can avoid these singular configurations.} 

The free part $\tilde S_0$ of the action, which is of second order in $(A, \psi)$, does not give rise to hyperbolic equations of motion, i.e., a well-posed initial value formulation. To remedy this, we follow the standard procedure (see, for example, \cite{HollandsYM, RejznerBook, DuetschBook}) and first introduce Lagrange multipliers $B$, ghosts $C$, and anti-ghosts $\tilde C$, all of which are sections of $\p$, and define the Batalin-Vilkovisky (BV) operator $s$ by (the commutators are Lie algebra commutators and $\bar \nabla$ the covariant derivative induced by the background connection $\bar \cA$)\footnote{With a slight abuse of notation, we here denote the Lie algebra representation on the matter fields by the same symbol as the group representation $\rho$ from which it is derived.}
\beq
\label{eq:s}
 s A^I_\mu = \bar \nabla_\mu C^I + [A_\mu, C]^I, \quad
 s C^I = - \frac{1}{2} [C, C]^I, \quad
 s \tilde C^I = B^I, \quad
 s B^I = 0, \quad
 s \psi = - \rho(C) \psi.
\eeq
We collectively denote these fields by $\Phi = (A_\mu^I, B^I, C^I, \tilde C^I, \psi)$ and attribute a ghost number $(0,0,1,-1,0)$ and mass dimension $(1,2,0,2, *)$, with the mass dimension of the matter fields the canonical one ($1$ for scalars, $\frac{3}{2}$ for spin $\frac{1}{2}$ fermions). Then $s$ increases the ghost number by one, but does not increase the mass dimension. The Grassmann parity of the fields is the same as their ghost number modulo 2, except for fermionic matter fields, which have Grassmann parity $1$. The BV differential is nil-potent when graded \wrt the Grassmann parity.

We also introduce anti-fields $\Phi^\ddag = (A^{\mu I \ddag}, B^{I \ddag}, C^{I \ddag}, \tilde C^{I \ddag}, \psi^\ddag)$, which we interpret as densities with ghost numbers $(-1, -1, -2, 0, -1)$. We supplement the action with a source term
\beq
 S_{\mathrm{source}} = - \int \sum_i s \Phi^i \Phi_i^\ddag,
\eeq
with $i$ labelling the different fields. With this definition $s F = (S_{\mathrm{source}}, F)$
for any local functional $F$ not depending on anti-fields, with the anti-bracket of local functionals defined by
\beq
 (F, G) \defeq \int \left\{ \frac{\delta^R F}{\delta \Phi^i(x)} \frac{\delta^L G}{\delta \Phi^\ddag_i(x)} - \frac{\delta^R F}{\delta \Phi^\ddag_i(x)} \frac{\delta^L G}{\delta \Phi^i(x)} \right\}.
\eeq
The anti-bracket fulfills a graded symmetry and a graded Jacobi identity, see \cite{TehraniBRST}, for example.

Gauge fixing is performed by choosing a gauge-fixing fermion $\Psi$, which we take to be of the form
\beq
\label{eq:Xi}
 \Psi = \int \tilde C^I \left( \bar \nabla^\mu A^I_\mu + \tfrac{1}{2} B^I \right) \vol.
\eeq
We then supplement the action by an additional gauge fixing term $s \Psi$.
The new action $S = \tilde S + S_{\mathrm{source}} + s \Psi$ thus obtained consists of a free part $S_0$ containing all terms of second order in (anti-) fields, and an interacting part $S_\ia$ containing all terms of higher order in (anti-) fields.\footnote{Note that we drop the terms of $\tilde S$ of order lower than 2 in the fields (the terms of first order in the fields are anyway a total derivative by our assumption that the background connection is on-shell).} Then we extend the definition of $s$ to functionals involving anti-fields by
\beq
\label{eq:s_F_(S,F)}
 s F = (S, F).
\eeq

The original action depended on $\bar \cA$, $A$ only through the combination $\bar \cA + A$, and was in this sense background independent. Analogously, we may say that a classical local functional $F$ is background independent, if it only depends on this combination, or in other words, if for any $\ba \in T \sS$,
\beq
\label{eq:cD_F_0}
 \cD_{\ba} F \defeq \bar \delta_{\ba} F - \delta_{\ba} F = 0,
\eeq
where $\bar \delta_{\ba}$ denotes the functional derivative \wrt the background field $\bar \cA$ in the direction $\ba$, defined as
\beq
  \bar \delta_\ba F[\bar \cA, \Phi, \Phi^\ddag] \defeq \frac{\ud}{\ud \tau} F[\bar \cA + \tau \ba, \Phi, \Phi^\ddag]|_{\tau = 0},
\eeq
while $\delta_{\ba}$ denotes the functional derivative \wrt the dynamical field $A$, in the same direction. The gauge fixing destroys the background independence of the action, as we have the formal\footnote{This is only formal as the integrals will in general not be convergent. This will be remedied below, in \eqref{eq:BrokenBackgroundIndependenceCutoff}, after we introduce an adiabatic cutoff for the interaction.} identity \cite{BgInd}
\beq
\label{eq:BrokenBackgroundIndependence}
 \bar \delta_{\ba} S - \delta_{\ba} S_\ia = s \bar \delta_{\ba} \Psi.
\eeq
The appropriate replacement of the condition \eqref{eq:cD_F_0} for a classical local functional $F$ to be background independent is then \cite{BgInd}
\beq
 \hat \cD_{\ba} F \defeq \cD_{\ba} F - (F, \cD_{\ba} \Psi) = 0.
\eeq
Note that, as $\cD_\ba \Psi$ does not contain anti-fields, $\hat \cD_\ba$ and $\cD_\ba$ coincide on functionals $F$ not containing anti-fields.
From now on, we think of $\ba$ not as a vector, i.e., an element of $T \sS$, but as a vector field, i.e., a section of $T \sS$. The operator $\hat \cD$ then fulfills 
\cite[Thm.~3.3]{BgInd}
\begin{align}
 \hat \cD_{\ba} (F_1, F_2) & =  ( \hat \cD_{\ba} F_1, F_2) +  (F_1, \hat \cD_{\ba} F_2), \\
\label{eq:s_cD}
 s \hat \cD_{\ba} F & = \hat \cD_{\ba} s F, \\
\label{eq:cD_cD}
 \hat \cD_{\ba} \hat \cD_{\ba'} F - \hat \cD_{\ba'} \hat \cD_{\ba} F & = \hat \cD_{\LB{\ba}{\ba'}} F,
\end{align}
where $\LB{ \cdot }{ \cdot }$ denotes the Lie bracket of vector fields. In view of the last equality, it is natural to interpret $\hat \cD_{\ba} F = (\hat \cD F)(\ba)$ as the evaluation of the ``background variation one-form'' $\hat \cD F$ in the vector field $\ba$. We can thus promote $\hat \cD$ to a differential, which, by \eqref{eq:s_cD}, anticommutes with $s$ (if the grading of background variation forms is taken into account). Put differently, one treats the perturbation $\ba$ as a new fermionic field variable.\footnote{A similar approach was used in \cite{KlubergSternZuber1}, see also Rem.~3.5 of \cite{BgInd} for further references.} This is discussed in more detail in Section~\ref{sec:Proof} below.

From the free part $S_0$ of the action, one obtains hyperbolic equations of motions for the fields $\Phi^i$, giving rise to a causal propagator $\Delta^{ij}$ being the difference of retarded and advanced propagator. One then defines a $\star$ product of functionals\footnote{Technically, one restricts to \emph{microcausal functionals}, which contain the local functionals and close under the $\star$ product \cite{DuetschFredenhagenLoopExpansion, DuetschFredenhagenDeformation}.} such that
\beq
 \Phi^i(x) \star \Phi^j(x') - (-1)^{\eps_i \eps_j} \Phi^j(x') \star \Phi^i(x) = i \hbar \Delta^{ij}(x, x'),
\eeq
with $\eps_i$ the Grassmann parity of $\Phi^i$. The anti-fields (anti-) commute with all other fields \wrt the $\star$ product. The $\star$ product preserves the grading $\Deg = 2 \deg_\hbar + \deg_{\Phi^{(\ddag)}}$ with $\deg_\hbar$ counting the factors of $\hbar$ and $\deg_{\Phi^{(\ddag)}}$ the total number of fields and anti-fields. The $\Deg$ grading can be used to make sense of the series typically occurring in perturbation theory, i.e., equations relating such series can be understood as equations at each order in $\Deg$.\footnote{This grading also naturally occurs in the context of deformation quantization, \cf \cite{WaldmannBook}.}

As in \cite{HollandsWaldTO, HollandsYM}, one defines local and covariant time-ordered products $T(F_1 \otimes \dots \otimes F_n)$ of local functionals $F_j$, which can be inductively fixed up to the total diagonal by their defining property of causal factorization, i.e.,
\beq
\label{eq:CausalFactorization}
 T(F_1 \otimes \dots \otimes F_n) = T(F_1 \otimes \dots \otimes F_k) \star T(F_{k+1} \otimes \dots \otimes F_n)
\eeq
whenever  $\supp F_i$ does not intersect the causal past of $\supp F_j$ for all $1\leq i \leq k$ and $k+1 \leq j \leq n$,
and which in particular are graded symmetric and fulfill field independence
\beq
\label{eq:T_FieldIndependence}
 \tfrac{\delta}{\delta \Phi(x)} T(F_1 \otimes \dots \otimes F_n) = \sum_{k = 1}^n T(F_1 \otimes \dots \otimes \tfrac{\delta}{\delta \Phi(x)} F_k \otimes \dots \otimes F_n).
\eeq
Apart from the axioms listed in \cite{HollandsWaldTO, HollandsYM} (we refrain from listing and explaining all the axioms, instead focusing on those that are relevant for our considerations) we also require that a time-ordered product with a single field factor simplifies as (for notational convenience, we assume the $F_k$ to be Grassmann even) \cite{HollandsWaldStress}
\begin{align}
 T( \Phi^i(x) \otimes F_1 \otimes \dots \otimes F_n ) & = \Phi^i(x) \star T( F_1 \otimes \dots \otimes F_n ) \nn \\
 & \quad + i \hbar \int \Delta_-^{ij}(x, y) \sum_{k =1}^n T( F_1 \otimes \dots \otimes \tfrac{\delta}{\delta \Phi^j(y)} F_k \otimes \dots \otimes F_n),
\end{align}
where $\Delta_-$ is the advanced propagator and integration over $y$ is understood (recall that we interpret $\frac{\delta}{\delta \Phi^j(y)} F$ as a density). Furthermore, time-ordered products respect the $\Deg$ grading and can be assumed to act trivially on anti-fields (here $\cO$ is an arbitrary local field):\footnote{As the anti-fields (anti-) commute with all other fields, this assumption is consistent with the other axioms. As it is not stated explicitly in \cite{HollandsYM}, we do so here, even though it is usually (implicitly) assumed in the literature.} 
\beq
\label{eq:T_trivial_on_antifields}
 T( \Phi^\ddag_i(x) \cO(x) \otimes F_1 \otimes \dots \otimes F_n) = \Phi^\ddag_i(x) T( \cO(x) \otimes F_1 \otimes \dots \otimes F_n). 
\eeq
As shown in \cite{HollandsWaldTO, BrunettiFredenhagenScalingDegree}, time-ordered products complying to these requirements can be constructed recursively, using the Epstein-Glaser method \cite{EpsteinGlaser} (extension of distributions to the diagonal).

Time-ordered products are not unique, but are subject to renormalization ambiguities encoded in the \emph{main theorem of renormalization} \cite{PopineauStora, HollandsWaldRG}, which is most conveniently formulated in terms of generating functionals
(equations involving such generating functionals have to be understood as yielding identities order by order in $F$)
\beq
 T(\eti{F}) \defeq 1 + \sum_{n=1}^\infty \frac{(\ibar)^n}{n!} T(F^{\otimes n}).
\eeq
It states that any two time ordered products $T$, $T'$ are related by
\beq
\label{eq:RenormalizationChange}
 T'(\eti{F}) = T(\eti{(F + D(\et{F}))})
\eeq
where $D(F_1 \otimes \dots \otimes F_n)$ are local functionals (supported on the total diagonal of the $F_k$'s), and at least of $\cO(\hbar)$.\footnote{Note that we do not assume that $D(F) = 0$.} They inherit many of the properties of time-ordered products, such as field independence, i.e., \eqref{eq:T_FieldIndependence} also holds for $T$ symbols replaced by $D$ symbols. With respect to the $\Deg$ grading, they fulfill
\beq
\label{eq:Deg_T_c}
 \Deg D(F_1 \otimes \dots \otimes F_n) = \sum_{i = 1}^n \Deg (F_i) - 2 (n-1).
\eeq
Furthermore, if $F_k$ are integrated fields of mass dimension $d_k$, then $D(F_1 \otimes \dots \otimes F_n)$ is an integrated field of mass dimension $4 + \sum_{k} (d_k - 4)$ (assuming that the test sections used for the smearing are attributed vanishing mass dimension).

It is sometimes also convenient to consider connected time-ordered products, which can be implicitly defined by\footnote{Note that this definition deviates from the one employed in \cite{DuetschFredenhagenLoopExpansion}, where the classical product replaces the $\star$ product.}  \cite{FrobBV}
\beq
 T(\eti{F}) = \exp_\star( \tfrac{i}{\hbar} T_c( \et{F} )),
\eeq
where
\beq
 \exp_\star(G) = 1 + \sum_{n=1}^\infty \frac{1}{n!} \underbrace{ G \star \dots \star G }_{n \text{ factors }}.
\eeq
One advantage of connected time-ordered products (which is also the reason why we consider them in the following) is that they are formal power series in $\hbar$ \cite[Thm.~1]{FrobBV} (in contrast to $T(\eti{F}$)). With respect to the $\Deg$ grading, they fulfill the same relation \eqref{eq:Deg_T_c} as the renormalization map $D$.

The \emph{free BV differential} $s_0$ is the linear part of $s$, i.e., it acts on fields as
\begin{align}
 s_0 A^I_\mu & = \bar \nabla_\mu C^I, &
 s_0 C^I & = 0, &
 s_0 \tilde C^I & = B^I, &
 s_0 B^I & = 0, &
 s_0 \psi & = 0,
\end{align}
and more generally on local functionals by $s_0 F = (S_0, F)$. Like $s$, it is nilpotent.
The $\star$ product can be defined such that $s_0$ fulfills a graded Leibniz rule \wrt to it \cite{HollandsYM, BgInd}.
The incompatibility of time-ordered products with $s_0$ is encoded in the anomaly, through the \emph{anomalous Ward identity} \cite[Props.~2,~4,~5]{HollandsYM}:
\begin{theorem}
We have\footnote{In \cite{HollandsYM} the anomaly is defined via time-ordered products instead of connected time-ordered products. These definitions can easily be seen to be equivalent.}
\beq
 s_0 T_c(\et{F}) = T_c( \{ s_0 F + \tfrac{1}{2} (F,F) + A(\et{F}) \} \otimes \et{F} ).
\eeq
The anomaly $A$ maps tensor products of local functionals linearly to locals functionals. It increases the ghost number by one and is at least of $\cO(\hbar)$.
It is subject to the consistency condition
\beq
 s_0 A(\et{F}) + (F, A(\et{F})) + A( \{ s_0 F + \tfrac{1}{2} (F, F) + A(\et{F}) \} \otimes \et{F}) = 0.
\eeq
In particular, if $A^m(\et{F})$ is the first non-trivial term in the $\hbar$ expansion of $A(\et{F})$, then
\beq
\label{eq:ExpandedConsistencyCondition}
 s_0 A^m(\et{F}) + (F, A^m(\et{F})) + A^m( \{ s_0 F + \tfrac{1}{2} (F, F) \} \otimes \et{F}) = 0.
\eeq
If $T$ and $T'$ are two renormalization schemes related by \eqref{eq:RenormalizationChange}, the corresponding anomalies $A$ and $A'$ are related by
\begin{multline}
\label{eq:AnomalyRedefinition}
 s_0 D(\et{F}) + (F, D(\et{F})) + \tfrac{1}{2} (D(\et{F}), D(\et{F})) + A(\et{F + D(\et{F})})\\ = D( \{ s_0 F + \tfrac{1}{2} (F, F) + A'(\et{F}) \} \otimes \et{F}) + A'(\et{F}).
\end{multline}
\end{theorem}

The anomaly $A$ inherits many of the properties of time-ordered products. In particular, it is field (and anti-field) independent,
\beq
\label{eq:AnomalyFieldIndependence}
  \tfrac{\delta}{\delta \Phi(x)} A(F_1 \otimes \dots \otimes F_n) = \sum_{k = 1}^n A(F_1 \otimes \dots \otimes \tfrac{\delta}{\delta \Phi(x)} F_k \otimes \dots \otimes F_n),
\eeq
vanishes if one of the factors is a linear field \cite[App.~A]{BgInd}, and fulfills the same relation \eqref{eq:Deg_T_c} to the $\Deg$ grading as the renormalization map. 
Also its relation to the mass dimension is as described for the renormalization map $D$ below \eqref{eq:Deg_T_c} \cite[Prop.~2]{HollandsYM}.

In order to define the algebra of interacting fields for a space-time region $\cR$, one first introduces a compactly supported adiabatic cutoff $\lambda(x)$ with $\lambda|_\cR = 1$ in the interacting part $S_{\ia}$ of the action. We do this in such a way that cubic terms of the Lagrangian are multiplied with $\lambda(x)$ while the quartic ones are multiplied with $\lambda^2(x)$. The generating functional for interacting local functionals $F$ supported inside $\cR$ is then (the inverse on the \rhs is \wrt $\star$)
\beq
 T^\ia(\eti{F}) = T(\eti{S_\ia})^{-1} \star T(\eti{F} \otimes \eti{S_\ia}),
\eeq
and these generate the interacting algebra for the space-time region $\cR$. The algebra thus obtained is independent on the adiabatic cutoff $\lambda$ (up to unitary equivalence) \cite{BrunettiFredenhagenScalingDegree}.

The interacting fields thus obtained are in general not gauge invariant, i.e., observables. In the classical theory, the gauge invariant, on-shell observables are given by the cohomology of $s$ at vanishing ghost number. In the quantum theory, there are quantum corrections due to anomalies. The \emph{quantum BV differential} \cite{HollandsYM, TehraniBRST, FrobBV}
\beq
 q F \defeq s F + A( F \otimes \et{S_\ia}),
\eeq
turns out to be nil-potent, i.e., a proper differential, in case that 
\beq
\label{eq:AbsenceGaugeAnomaly}
 \supp A(\et{S_\ia}) \cap \cR = \emptyset,
\eeq
i.e., the \emph{gauge anomaly} $A(\et{S_\ia})$ vanishes in the region where the interacting observables are considered.\footnote{In \cite{HollandsYM, TehraniBRST, FrobBV}, it is required that $A(\et{S_\ia})=0$, but it is easy to see that \eqref{eq:AbsenceGaugeAnomaly} is sufficient.}
In this case, we say that gauge anomalies are absent. Then the cohomology of $q$ determines the interacting local observables in the quantum theory. It is thus crucial to understand the potential obstructions to achieving \eqref{eq:AbsenceGaugeAnomaly}. 

The gauge anomaly $A(\et{S_\ia})$ is a local functional, which can be expressed as
\beq
\label{eq:AnomalyExpansion}
 \sum_{n=1}^\infty \sum_\alpha \int a^n_{\alpha_1 \dots \alpha_n} \nabla^{\alpha_1} \lambda \dots \nabla^{\alpha_n} \lambda \vol,
\eeq
with $\alpha_k$ multiindices and $a^n_{\alpha_1 \dots \alpha_n}$ determined locally and covariantly out of (anti-) fields and background fields and symmetric under interchange $\alpha_i \leftrightarrow \alpha_j$. The above expression can be understood in the sense of a formal power series in $\lambda$. However, with our choice of adiabatic cutoff in $S_\ia$, the expansion in powers of $\lambda$ is equivalent to an expansion in the Deg grading, as
$\Deg(a^n_{\alpha_1 \dots \alpha_n}) = n + 2$.
We have the following:

\begin{lemma}
\label{lemma:field}
An expression of the form \eqref{eq:AnomalyExpansion} with
\beq
\label{eq:Deg_n_relation}
 \Deg(a^n_{\alpha_1 \dots \alpha_n}) = n + c
\eeq
for some $c \in \N$
vanishes for all compactly supported $\lambda$ (order by order in $\Deg$) only if, for all $n$,
\beq
\label{eq:a_n_nabla_b}
 a^n_{\emptyset \dots \emptyset} = \nabla^\mu b^n_\mu
\eeq
for some local field $b^n_\mu$.
Hence, an expression of the form \eqref{eq:AnomalyExpansion} determines a four-form
\beq
\label{eq:aDef}
 a = \sum_{n=1}^\infty a^n_{\emptyset \dots \emptyset} \vol
\eeq
up to a total derivative, i.e., modulo $\ud$, in the following sense: Two expression of the form \eqref{eq:AnomalyExpansion} coincide for all compactly supported $\lambda$ if and only if the corresponding four-forms $a$ defined by \eqref{eq:aDef} coincide modulo an exact form.
\end{lemma}

\begin{proof}
If \eqref{eq:AnomalyExpansion} vanishes for all compactly supported $\lambda$, then the functional derivative \wrt $\lambda(x)$ must vanish, i.e.,
\beq
 \sum_{n=1}^\infty \sum_\alpha \sum_{i = 1}^n (-1)^{| \alpha_i |} \nabla^{\bar \alpha_i} \left( a^n_{\alpha_1 \dots \alpha_n} \nabla^{\alpha_1} \lambda \dots \cancel{ \nabla^{\alpha_i} \lambda } \dots \nabla^{\alpha_n} \lambda \right)(x) =0,
\eeq
where $\bar \alpha_i$ is the multi-index in reverse order. This must hold order by order in the $\Deg$ grading, so we can restrict to a fixed $n$. Choosing $\lambda = 1$ near $x$, we obtain
\beq
 \sum_k (-1)^k \nabla^{\mu_k \dots \mu_1} a^n_{(\mu_1 \dots \mu_k) \emptyset \dots \emptyset} = 0.
\eeq
All but the first ($k = 0$) term on the \lhs are total derivatives, proving the statement.
\end{proof}

Expressions of the form \eqref{eq:AnomalyExpansion}, with coefficients $a^n_{\alpha_1 \dots \alpha_n}$ fulfilling a relation \eqref{eq:Deg_n_relation},
will frequently occur, not just for the gauge anomaly. Also to such expressions we can associate a four form (up to a total derivative). We denote this association by $ \cdot |_{\lambda = 1}$, i.e., in the above case\footnote{Note that $\cdot |_{\lambda = 1}$ does not stand for setting $\lambda = 1$ in the (integrated) local functional. Instead, it associates a four form (defined up to a total derivative) to the local functional.}
\beq
 \left. \sum_{n=1}^\infty \sum_\alpha \int a^n_{\alpha_1 \dots \alpha_n} \nabla^{\alpha_1} \lambda \dots \nabla^{\alpha_n} \lambda \vol \right|_{\lambda = 1} = a.
\eeq
We call two expressions $F$, $G$, of the form \eqref{eq:AnomalyExpansion} equivalent, $F \sim G$, if $F|_{\lambda = 1} = G|_{\lambda = 1} \mod \ud$
(this is similar to the equivalence of ``generalized Lagrangians'' used in \cite{BDF09}). We can thus reformulate the condition \eqref{eq:AbsenceGaugeAnomaly} as
\beq
\label{eq:AbsenceGaugeAnomalySim}
 A(\et{S_\ia}) \sim 0.
\eeq

Now denote by $A^k(\et{S_\ia})$ the $\cO(\hbar^k)$ component of $A(\et{S_\ia})$ and assume that $A^k(\et{S_\ia}) \sim 0$ for all $k < m$, while $A^m(\et{S_\ia}) \not\sim 0$, i.e., $A^m(\et{S_\ia})$ is the first non-trivial term in the $\hbar$ expansion of $A(\et{S_\ia})$. Denote by $a^m \defeq A^m(\et{S_\ia})|_{\lambda = 1}$ the corresponding four-form according to \eqref{eq:aDef}. Then, using that\footnote{The nilpotency of $s$ and $s_0$ leads, via the Jacobi identity for the antibracket, to the formal identity $2 s_0 S_\ia + (S_\ia, S_\ia)  = 0$. Including the adiabatic cutoff in $S_\ia$ remedies the formality, but introduces violation terms involving $\nabla \lambda$, so that the equality is reduced to the $\sim$ equivalence.}
\beq
 s_0 S_\ia + \tfrac{1}{2} (S_\ia, S_\ia) \sim 0,
\eeq 
it follows from \eqref{eq:ExpandedConsistencyCondition} and the fact that $[s_0 A^m(\et{S_\ia}) + (S_\ia, A^m(\et{S_\ia}))] |_{\lambda = 1} = s a^m \mod \ud$, that $a^m$ is $s$ closed modulo $\ud$, i.e.,
\beq
 s a^m = \ud b^m.
\eeq
On the other hand, if $a^m$ is $s$ exact modulo $\ud$, i.e., $a^m = s c^m + \ud d^m$, then one can perform a redefinition of time-ordered products as in \eqref{eq:RenormalizationChange} to obtain $a^m = 0$ \cite{HollandsYM}, see also \eqref{eq:AnomalyRemoval} below.
Proceeding inductively in the order of $\hbar$, one obtains absence of the gauge anomaly, in case that each $a^m$ is $s$ exact modulo $\ud$. Hence, the potential obstructions to achieve \eqref{eq:AbsenceGaugeAnomalySim} are characterized by the cohomology $H^1_4(s | \ud )$ of local and covariant four-forms of ghost number $1$ and mass dimension $4$ (or less, if the matter action contains terms of mass dimension less than $4$). The most general such element of $H^1_4(s | \ud )$ is a linear combination of the \emph{non-abelian anomaly}
\begin{multline}
\label{eq:NonabelianAnomaly}
 \Tr_N \left[ C \left( \bar F \bar F + \tfrac{1}{2} ( \bar \nabla A \bar F + \bar F \bar \nabla A ) + \tfrac{1}{6} ( A^2 \bar F - 2 A \bar F A + \bar F A^2 ) \right. \right. \\
  \left. \left. + \tfrac{1}{3} \bar \nabla A \bar \nabla A + \tfrac{1}{6} ( \bar \nabla A A^2 - A \bar \nabla A A + A^2 \bar \nabla A ) \right) \right],
\end{multline}
with $\Tr_N$ an invariant trace over one of the non-abelian factors and $\bar \nabla$ the covariant differential\footnote{It is not a proper differential, as it squares to the commutator with the curvature two-form $\bar F$ of the background connection.} defined by $\bar \nabla A = \bar \nabla_\mu A_\nu \ud x^\mu \ud x^\nu$, and an \emph{abelian anomaly}
\beq
\label{eq:AbelianAnomaly}
 C^A G^A 
\eeq
with $C^A$ the abelian ghost and $G^A \in H^0_4(s)$, \cf \cite{BarnichBrandtHenneaux00, ManesStoraZumino}. 
For the later considerations, it is important to note two facts about the above anomalies. First, the non-abelian anomaly \eqref{eq:NonabelianAnomaly} vanishes if and only if the $d$ symbol (symmetric structure constant) of the trace $\Tr_N$ vanishes, i.e., $\Tr_N (T^I \{ T^J, T^K \}) = 0$ for generators $T^I$ of the gauge group $G$. Second, for an abelian anomaly \eqref{eq:AbelianAnomaly}, a representer $G^A$ of a nontrivial cohomology class in $H^0_4(s)$ can be chosen such that if $G^A_i$ is the lowest non-trivial component in a filtration \wrt the total (anti-)field number, then $G^A_i$ is a non-trivial element of $H^0_4(s_0)$, so in particular not $s_0$ exact, see \cite[Prop.~5.6]{PiguetSorella}. Furthermore, the restriction of $G^A_i$ to the component of vanishing anti-field number does not vanish, i.e. $G^A_i|_{\Phi^\ddag = 0} \neq 0$, \cite[Thm~7.1]{BarnichBrandtHenneaux00}. It follows that $G^A_i \not \approx_0 0$, where $\approx_0$ denotes equality modulo the free equations of motions (generated by $S_0$), when anti-fields are set to zero. For if $G^A_i \approx_0 0$, then $G^A_i|_{\Phi^\ddag = 0} = (s_0 H_i)|_{\Phi^\ddag = 0}$ for some $H_i$. Then $G^{\prime A}_i \defeq G^A_i - s_0 H_i$ is in the same equivalence class of $H^0_4(s_0)$ as $G^A_i$, but has $G^{\prime A}_i|_{\Phi^\ddag = 0} = 0$.

\begin{remark}
\label{rem:AbelianAnomaly}
Examples for abelian anomalies are
\begin{align}
\label{eq:AbelianChiral}
 G^A & = F^A F^A, &
 G^A & = \Tr_N F F, &
 G^A & = \Tr_{so(3,1)} R R,
\end{align}
with $R$ the Riemann curvature tensor, interpreted as an $so(3,1)$ valued two-form.
These are relevant in particular for the Standard Model (the vanishing of the corresponding coefficient at $\cO(\hbar)$ essentially fixes the hypercharges of quark and leptons \cite{GengMarshak}). Other potential abelian anomalies would be
\begin{align}
\label{eq:AbelianNonChiral}
 G^A & = F^A_{\mu \nu} F^{A \mu \nu} \vol, &
 G^A & = \Tr_N F_{\mu \nu} F^{\mu \nu} \vol , &
 G^A & = R_{\mu \nu \lambda \rho} R^{\mu \nu \lambda \rho} \vol, &
 G^A & = \skal{\vp}{\vp}_V^2 \vol.
\end{align}
Here $V$ is the representation space of the bosonic matter field $\vp$ and $\skal{ \cdot }{ \cdot }_V$ a $G$ invariant hermitean inner product. Obviously, there are many possible variants of the last term, for example involving the fermionic matter fields or derivatives. Such anomalies have, up to now, not occurred in perturbative calculations and are scarcely discussed in the literature (an exception is \cite{Nakayama2018}). However, we are not aware of a general argument excluding such anomalies.\footnote{Note that the anomalies in \eqref{eq:AbelianChiral} have parity opposite to those in \eqref{eq:AbelianNonChiral}. If there is a general argument ruling out anomalies of the form \eqref{eq:AbelianNonChiral}, this may be expected to be important.} We intend to revisit this topic in the future. In any case, the proof of our main theorem given below does apply to all potential abelian anomalies.
\end{remark}

In order to formulate the requirement of perturbative agreement, we have to relax the condition that the background is on-shell, so that we are able to consider infinitesimal changes of the background connection which are compactly supported (and not necessarily pure gauge). We can simply do this by extending the action $S_0 + S_\ia$ to off-shell backgrounds.\footnote{Note that this is not the same as the expansion of the original action around off-shell backgrounds, as this would lead to a term linear in the dynamical fields, which we do not consider.} The free BV differential $s_0$ will then no longer be nil-potent, except when acting on functionals supported in the region where the background is on-shell. We denote by $a$ a compactly supported section of $\p \otimes \Omega^1$, representing an infinitesimal variation of the background connection. Perturbative agreement \wrt changes in the background gauge now means that, for any such $a$, \cite{HollandsWaldStress, BackgroundIndependence}
\beq
\label{eq:PA_def}
 \delta^\ret_a T(\eti{F}) = T( \tfrac{i}{\hbar} \{ \bar \delta_a F + \bar \delta_a S_0 \} \otimes \eti{F} ) - T(\tfrac{i}{\hbar} \bar \delta_a S_0) \star T( \eti{F} ).
\eeq
It formalizes the notion that it should not matter whether we quantize around a background connection $\bar \cA'$ or around an (infinitesimally close) background connection $\bar \cA$, with the difference to $\bar \cA'$ taken into account in an interaction term $S_0[\bar \cA'] - S_0[\bar \cA]$. The identification of the algebras over different backgrounds is implemented via the \emph{retarded variation} $\delta^\ret_a$ \wrt the infinitesimal background perturbation $a$ (as it is not used in the proof of our main statement, we refrain from stating the precise definition, which can be found for example in \cite{BgInd}).
Perturbative agreement is a consistency condition between renormalization prescriptions on different backgrounds, beyond the constraints imposed by local covariance.
The quantity $J_0(a) = \bar \delta_a S_0$ occurring in \eqref{eq:PA_def} can be seen as the free current $J_0$ smeared with the compactly supported Lie-algebra valued one-form $a$. By induction in the total number of fields, one can show \cite{HollandsWaldStress, BackgroundIndependence} that perturbative agreement can be fulfilled provided that
\beq
\label{eq:E}
 E(a_1, a_2) \defeq \delta^\ret_{a_1} T(J_0(a_2)) - \delta^\ret_{a_2} T(J_0(a_1)) + \tfrac{i}{\hbar} [ T(J_0(a_1)), T(J_0(a_2))]_\star = 0,
\eeq
a condition which is in fact closely related to the Wess-Zumino consistency condition \cite{WessZuminoConsistencyCondition}, see \cite{GlobalAnomalies}. In dimensions $D \leq 4$, this condition can in turn be fulfilled\footnote{Using the explicit form of $\delta^\ret_a$ given in \cite{BgInd} one can easily check that \eqref{eq:T_trivial_on_antifields}, which was not required in \cite{BgInd}, can be preserved under the necessary redefinitions of time-ordered products.} \cite{BackgroundIndependence} provided that the covariant divergence of the free current vanishes,
\beq
\label{eq:J_0_nabla_Lambda}
 T(J_0(\bar \nabla \Lambda)) \approx_0 0,
\eeq
for any compactly supported section $\Lambda$ of $\p$ (a ``Lie algebra valued function''). Here $\approx_0$ denotes equality modulo the free equations of motion, when anti-fields are set to zero.
This condition can be rather straightforwardly checked in concrete cases, \cf \cite{ChiralFermions} for the case of chiral fermions in gauge backgrounds. The calculation boils down to determining the violation of the Hadamard parametrix to being a solution of the free field equation, a calculation which actually does not involve any loop integrals. For our purposes the following is crucial \cite[Prop~3.4]{BackgroundIndependence} \cite[Thm.~3.8]{BgInd}:

\begin{theorem}
\label{thm:PA}
If \eqref{eq:J_0_nabla_Lambda} and perturbative agreement \wrt changes in the background gauge connection is fulfilled, then
\begin{align}
\label{eq:PA_Thm_1}
 T( \bar \delta_{\bar \nabla \Lambda} S \otimes \eti{S_\ia} ) & \approx_0 0, \\
\label{eq:PA_Thm_2}
 \bar \delta_{ \ba} A( \et{F} ) & = A( \{ \bar \delta_{\lambda \ba} S_0 + \bar \delta_\ba F \} \otimes \et{F} ).
\end{align}
In \eqref{eq:PA_Thm_1}, $\Lambda$ is supported in the region where $\lambda = 1$ and $\approx_0$ denotes equality modulo the free equations of motion (derived from $S_0$), when anti-fields are set to zero. In \eqref{eq:PA_Thm_2}, $F$ is supported in the region where $\lambda = 1$, and $\ba$ is again a solution to the linearized equations of motions, i.e., an infinitesimal variation of an on-shell background connection.
\end{theorem}

Note that in \eqref{eq:PA_Thm_2}, we cut off the (not necessarily compactly supported) infinitesimal variation $\ba$ of the free action $S_0$ by multiplication with the adiabatic cut-off $\lambda$. This is necessary to have the expression well-defined, but obviously the choice of $\lambda$ is irrelevant by the locality of the anomaly and the restriction on the support of $F$.\footnote{Note that $A( \bar \delta_{\lambda \ba} S_0) = 0$, as by the $\Deg$ grading rule the result must be a c number, which is not possible at ghost number 1.}
Later, we will apply \eqref{eq:PA_Thm_2} to the case $F = S_\ia$, whose support is not restricted to the region where $\lambda = 1$. It follows that there will be supplementary terms supported in $\supp \ud \lambda$, i.e., the equality in \eqref{eq:PA_Thm_2} is replaced by the equivalence $\sim$ as defined below Lemma~\ref{lemma:field}.

In the following, we will need to perform redefinitions of time-ordered products, but have to preserve perturbative agreement. For this, the following is essential (we omit the straightforward proof):

\begin{lemma}
\label{lemma:D_consistency}
If time-ordered products $T$ fulfill perturbative agreement, then the time-ordered products $T'$ defined by \eqref{eq:RenormalizationChange} also fulfill perturbative agreement if and only if
\beq
\label{eq:D_consistency}
 \bar \delta_a D(\et{F}) = D( \{ \bar \delta_a F + \bar \delta_a S_0 \} \otimes \et{F}).
\eeq
\end{lemma}

Finally, we note that in the presence of an adiabatic cutoff $\lambda$ in the interaction, the identity \eqref{eq:BrokenBackgroundIndependence} can be stated as follows:
\beq
\label{eq:BrokenBackgroundIndependenceCutoff}
 \bar \delta_{\lambda \ba} S - \delta_\ba S_\ia - (S, \bar \delta_{\lambda \ba} \Psi) \sim 0.
\eeq
Note the presence of the adiabatic cut-off in all functional derivatives \wrt background fields. It ensures the well-definedness of these expressions, but also that in all terms the same (affine) linear relation between the $\Deg$ grading and the order in $\lambda$ holds. 
Hence, the \lhs is of the form \eqref{eq:AnomalyExpansion}, with $\Deg( a^n_{\alpha_1 \dots \alpha_n} ) = n+1$, 
so that Lemma~\ref{lemma:field} is applicable.

\section{Proof of Theorem~\ref{thm:main}}
\label{sec:Proof}

As discussed above, if the conditions of Theorem~\ref{thm:main} are met, we can assume that perturbative agreement holds. We begin by arguing that trivial anomalies can be removed by redefinitions of time-ordered products, without destroying perturbative agreement. 

Assume $m$ is the lowest order in $\hbar$ at which the anomaly $a^m$ is not $\ud$ exact. If the anomaly is trivial in $H^1_4(s | \ud)$, i.e., 
\beq
\label{eq:a_m_trivial}
 a^m = s b^m + \ud c^m,
\eeq
then, following \cite{HollandsYM}, one performs a redefinition of the form \eqref{eq:RenormalizationChange} such that
\beq
\label{eq:AnomalyRemoval}
 D(\et{S_\ia}) |_{\lambda = 1} = \frac{1}{n!} D( S_1^{\otimes n}) |_{\lambda = 1} = - b^m.
\eeq
Here $S_1$ is the part of $S_\ia$ which is cubic in (anti-) fields and $n = 2 m - 2 + k$, with $k$ the total number of (anti-) fields in $b^m$ (the latter condition ensures that the redefined time-ordered products respect the $\Deg$ grading).\footnote{In case $b^m$ is not homogenous in the total number of (anti-) fields, one decomposes it and applies the above to each component.} By \eqref{eq:AnomalyRedefinition}, the anomaly at $\cO(\hbar^m)$ is then removed for the modified time-ordered products. 

Let us first verify that the modification \eqref{eq:AnomalyRemoval} can be made such that (anti-) field independence of time-ordered products is preserved (this was not explicitly shown in \cite{HollandsYM}). As $S_\ia$ contains neither the field $B^I$ nor the anti-fields $A^{A \ddag}_\mu$, $C^{A \ddag}$, $\tilde C^{I \ddag}$, $B^{I \ddag}$ (here $I$ is a general Lie algebra index, while $A$ stands for the abelian direction), we have to make sure that we can choose $b^m$ such that it also does not contain these. By field independence of the anomaly, we know that also $a^m$ is independent of these. We can thus use the following:

\begin{lemma}
\label{lemma:Antifields_in_b}
Let $a^m$ be a four form of ghost number one and mass dimension four, which is independent of $B^I$, $A^{A \ddag}_\mu$, $C^{A \ddag}$, $\tilde C^{I \ddag}$, $B^{I \ddag}$ and can be expressed as \eqref{eq:a_m_trivial}. Then $a^m$ can also be expressed as \eqref{eq:a_m_trivial} with the four form $b^m$ independent of $B^I$, $A^{A \ddag}_\mu$, $C^{A \ddag}$, $\tilde C^{I \ddag}$, $B^{I \ddag}$.
\end{lemma}

\begin{proof}
By the power counting and ghost number constraints, the only possible term in $b^m$ containing $A^{A \ddag}_\mu$ is (up to a total derivative) $\bar \nabla^\mu C^A A^{A \ddag}_\mu$. However, as $s (\bar \nabla^\mu C^A A^{A \ddag}_\mu) = - s ( A^{A\mu} s A^{A \ddag}_\mu )$, we can re-express it without the use of anti-fields. As $C^{A \ddag}$ is of ghost number $-2$ and there is no field of ghost number $2$ and vanishing mass dimension to multiply it with ($[C, C]^I$ has no component in the abelian direction), $C^{A \ddag}$ can not occur in $b^m$. Finally, up to total derivatives and elements in the kernel of $s$, any linear combination of terms containing at least one $B^I$, $\tilde C^{I \ddag}$, or $B^{I \ddag}$, and such that the image under $s$ does not contain any of these, can be written as
\beq
\label{eq:GeneralCombination}
 D^I ( B^I \vol - \tilde C^{I \ddag} ) + s D^I B^{I \ddag} 
\eeq
for some $D^I$ of vanishing ghost number (which in principle could still contain the undesired fields). However, the image of this under $s$ can be expressed as $- s ( D^I \bar \nabla^\mu A^I_\mu ) \vol$, so that we can replace \eqref{eq:GeneralCombination} in $b^m$ by an expression which contains one power of $B^I$, $\tilde C^{I \ddag}$, $B^{I \ddag}$ less. In case $D^I$ still contains the undesired fields, one iterates the procedure.
\end{proof}

Hence, we may choose $b^m$ in \eqref{eq:a_m_trivial} such that it does not contain any of the (anti-) fields $B^I$, $A^{A \ddag}_\mu$, $C^{A \ddag}$, $\tilde C^{I \ddag}$, $B^{I \ddag}$, so that the redefinition \eqref{eq:AnomalyRemoval} is not obstructed by (anti-) field independence. In order to also preserve \eqref{eq:T_trivial_on_antifields} under the field redefinition, we have to make sure that the total anti-field content on both sides of \eqref{eq:AnomalyRemoval} match. This can be achieved by decomposing $S_1$ into monomials and choosing the combination of these on the \lhs of \eqref{eq:AnomalyRemoval} such that total anti-field content matches that of $b^m$.

Finally, let us check that the redefinition \eqref{eq:AnomalyRemoval} does not spoil perturbative agreement. For this, we must have, by Lemma~\ref{lemma:D_consistency}, (all equalities modulo $\ud$)
\beq
 - \bar \delta_\ba b^m = D( \{ \bar \delta_{\lambda \ba} S_\ia + \bar \delta_{\lambda \ba} S_0 \} \otimes \et{S_\ia} )|_{\lambda = 1} = D( \bar \delta_{\lambda \ba} S \otimes \et{S_\ia} ) |_{\lambda = 1}.
\eeq
On the other hand, from field independence of $D$ and \eqref{eq:BrokenBackgroundIndependenceCutoff}, we also have (again modulo $\ud$)
\beq
 - \delta_\ba b^m = D( \delta_\ba S_\ia \otimes \et{S_\ia} )|_{\lambda = 1} = D( \{ \bar \delta_{\lambda \ba} S - (S, \bar \delta_{\lambda \ba} \Psi) \} \otimes \et{S_\ia} )|_{\lambda = 1}.
\eeq
These two redefinitions are linearly independent (and can thus both be performed independently), unless $(S, \bar \delta_{\lambda \ba} \Psi) \sim 0$. In that case the two redefinitions are consistent if and only if $\bar \delta_\ba b^m = \delta_\ba b^m \mod \ud$. Hence, we need to show that we can choose $b^m$ such that $(S, \bar \delta_{\lambda \ba} \Psi) \sim 0$ implies $\cD_\ba b^m = 0 \mod \ud$. 

For general perturbations $\ba$, we have, if the conditions of Theorem~\ref{thm:PA} are met, using \eqref{eq:PA_Thm_2}, field independence \eqref{eq:AnomalyFieldIndependence} of the anomaly, and the relation \eqref{eq:BrokenBackgroundIndependenceCutoff},
\beq
\label{eq:cD_A_A_S}
  \cD^\lambda_{\ba} A^m(\et{S_\ia}) - A^m((S, \bar \delta_{\lambda \ba} \Psi) \otimes \et{S_\ia} ) \sim 0,
\eeq
where we introduced $\cD^\lambda_\ba \defeq \bar \delta_{\lambda \ba} - \delta_\ba$.
Hence, $(S, \bar \delta_{\lambda \ba} \Psi) \sim 0$ implies that $\cD_\ba a^m = 0 \mod \ud$. 
Now $(S, \bar \delta_{\lambda \ba} \Psi) \sim 0$ only holds for a perturbation $\ba$ purely in the abelian direction, in which case $\bar \delta_{\lambda \ba} \Psi = 0$ and $\cD_\ba \Psi = - \int \tilde C^I \bar \nabla^\mu \ba^I_\mu \vol$. Hence, for such $\ba$, $(F, \cD_\ba \Psi) = 0$ unless $F$ contains $\tilde C^\ddag$. As already argued above, $a^m$ does not contain $\tilde C^\ddag$, so we also have $\hat \cD_\ba a^m = 0 \mod \ud$ for such $\ba$. Given a solution $b^m$, $c^m$ to $a^m = s b^m + \ud c^m$ for generic non-abelian background connections but a fixed abelian background, we may extend $b^m$, $c^m$ via ``parallel transport'' along the abelian direction \wrt the flat ``connection'' $\hat \cD$ (recall \eqref{eq:cD_cD}) to all backgrounds.\footnote{As the equations of motion for the abelian and the non-abelian background connections decouple, we can see the manifold $\sS$ of solutions to the Yang-Mills equation as a Cartesian product of the manifolds of abelian and non-abelian background connections.} Then $a^m = s b^m + \ud c^m$ is fulfilled on all backgrounds and $\hat \cD_\ba b^m = 0$ for $\ba$ in the abelian direction. As discussed above, by Lemma~\ref{lemma:Antifields_in_b}, we can choose $b^m$ such that it does not contain $\tilde C^\ddag$. Starting the parallel transport with such $b^m$, this will be the case on all backgrounds, so in particular the anti-bracket with $\cD_\ba \Psi$ will vanish for $\ba$ in the abelian direction. Hence, we also have $\cD_\ba b^m = 0$ for such $\ba$, so that there is no obstruction to preserve perturbative agreement in the removal of trivial gauge anomalies.

We may thus assume that at the first non-trivial order in $\hbar$ the anomaly is non-trivial, i.e., a linear combination of the non-abelian anomaly \eqref{eq:NonabelianAnomaly} and an abelian anomaly of the form \eqref{eq:AbelianAnomaly}. 

The above already showed that the abelian anomaly is background independent, i.e., $\cD_\ba a^m = 0$ (modulo $\ud)$ for an abelian anomaly $a^m$, \cf \eqref{eq:cD_A_A_S}. This is however not sufficient to rule it out, as there are many possible background independent abelian anomalies, for example the ones in \eqref{eq:AbelianChiral} with $F$ the full field strength (of $\cA = \bar \cA + A$). However, we can proceed differently:

\begin{proposition}
\label{prop:Abelian}
If \eqref{eq:J_0_nabla_Lambda} and perturbative agreement \wrt changes in the background connection holds, and $a^m = A^m(\et{S_\ia})|_{\lambda = 1}$ is the lowest non-trivial term in an $\hbar$ expansion of the gauge anomaly, then $a^m$ can not contain a contribution from a non-trivial abelian anomaly of the form \eqref{eq:AbelianAnomaly}.
\end{proposition}

\begin{proof}
If $a^m$ contains an abelian anomaly \eqref{eq:AbelianAnomaly}, then its coefficient $G^A$ can be obtained by functional differentiation \wrt the abelian ghost,
\beq
 G^A(x) = \tfrac{\delta}{\delta C^A(x)} A^m(\et{S_\ia}) = A^m(\tfrac{\delta}{\delta C^A(x)} S_\ia \otimes \et{S_\ia}),
\eeq
where in the second step we used field independence \eqref{eq:AnomalyFieldIndependence} of the anomaly. We may here and in the following restrict to $x$ contained in the region where the cutoff $\lambda$ equals $1$. On the other hand, by the definition of the anomaly, we have
\begin{multline}
\label{eq:s_0_del_C_gaugeAnomaly}
 s_0 T_c( \tfrac{\delta}{\delta C^A(x)} S_\ia \otimes \et{S_\ia} ) = T_c( (S, \tfrac{\delta}{\delta C^A(x)} S_\ia ) \otimes \et{S_\ia} ) + T_c(A( \tfrac{\delta}{\delta C^A(x)} S_\ia \otimes \et{S_\ia} ) \otimes \et{S_\ia} ) \\ 
  - T_c( \tfrac{\delta}{\delta C^A(x)} S_\ia \otimes \{ s_0 S_\ia + \tfrac{1}{2} (S_\ia, S_\ia) + A(\et{S_\ia
}) \} \otimes \et{S_\ia} ).
\end{multline}
We will be interested in considering this equality modulo anti-fields and the free equations of motion. The only terms in $S_\ia$ which depend on the abelian ghost are the source terms for the matter gauge transformations, i.e., schematically of the form $\int \rho(C) \psi \psi^\ddag $. Hence, $\tfrac{\delta}{\delta C^A(x)} S_\ia$ is linear in the matter anti-fields, and thus the last term on the \rhs of \eqref{eq:s_0_del_C_gaugeAnomaly} vanishes when anti-fields are set to zero (recall \eqref{eq:T_trivial_on_antifields}, i.e., that anti-fields are not contracted in time-ordered products). By the same argument, the \lhs of \eqref{eq:s_0_del_C_gaugeAnomaly} vanishes on-shell (\wrt to the free equations of motion) when anti-fields are set to zero (recall that the action of $s_0$ on anti-fields generates the free equations of motion). Regarding the first term on the \rhs of \eqref{eq:s_0_del_C_gaugeAnomaly}, we notice that, by the restriction of $x$ to the region where $\lambda$ equals $1$,
\begin{align}
 (S, \tfrac{\delta}{\delta C^A(x)} S_\ia ) & = (S_0,  \tfrac{\delta}{\delta C^A(x)} S_\ia ) - \tfrac{1}{2} \tfrac{\delta}{\delta C^A(x)} (S_\ia, S_\ia) \nn \\
 & = (S_0,  \tfrac{\delta}{\delta C^A(x)} S_\ia ) + \tfrac{\delta}{\delta C^A(x)} (S_0, S_\ia) \nn \\
 & = ( \tfrac{\delta}{\delta C^A(x)} S_0, S_\ia).
\end{align}
For convenience, we integrate against $\Lambda^A(x)$, i.e., we consider the functional derivative \wrt $C^A$ in the direction $\Lambda^A$ with $\Lambda^A$ supported within the region where $\lambda = 1$. Then
\beq
\label{eq:delta_C_S_0}
 \skal{\tfrac{\delta}{\delta C} S_0}{\Lambda} = - \int ( A^{\mu A \ddag} + \bar \nabla^\mu \bar C^A \vol ) \bar \nabla_\mu \Lambda^A.
\eeq
and thus
\beq
\label{eq:delta_S_0_Lambda_S_int}
 ( \skal{\tfrac{\delta}{\delta C} S_0}{\Lambda}, S_\ia ) = - \delta_{\bar \nabla \Lambda} S_\ia = - \bar \delta_{\bar \nabla \Lambda} S,
\eeq
where in the last step, we used \eqref{eq:BrokenBackgroundIndependenceCutoff} and the fact that the gauge fixing fermion does not depend on the abelian background gauge field. 
Hence, the first term on the \rhs of \eqref{eq:s_0_del_C_gaugeAnomaly} is ($-1$ times) the divergence of the interacting current, which, by \eqref{eq:PA_Thm_1},  vanishes on-shell \wrt the free equations of motion, when anti-fields are set to zero.
Hence, we have shown that 
\beq
 T_c( G^A(x) \otimes \et{S_\ia} ) \approx_0 0.
\eeq
Now assume that there is a non-trivial abelian anomaly $C^A G^A$ at the leading order $\hbar^m$. As the connected time-ordered products are formal power series in $\hbar$, the $\cO(\hbar^m)$ component of the \lhs can be expressed as $T^m_c( G^A(x) \otimes \et{S_\ia})$, where the superscript $m$ denotes the restriction to $\cO(\hbar^m)$. As each factor $S_\ia$ increases the total $\Deg$ of $T^m_c( G^A(x) \otimes \et{S_\ia})$ at least by $1$ (this follows from \eqref{eq:Deg_T_c} and $\Deg S_\ia \geq 3$), the lowest non-vanishing term in a filtration of $T^m_c( G^A(x) \otimes \et{S_\ia})$ \wrt the total (anti-) field number is $G^A_i(x)$ in the notation introduced below \eqref{eq:AbelianAnomaly}. Hence, we have shown that $G^A_i \approx_0 0$. But, as argued below \eqref{eq:AbelianAnomaly}, this is not possible for a non-trivial abelian anomaly.
\end{proof}

It remains to rule out the non-abelian anomaly  \eqref{eq:NonabelianAnomaly}.
As we already ruled out trivial and abelian anomalies, we can assume that the contribution $a^m$ of lowest non-vanishing order in $\hbar$ is of the form \eqref{eq:NonabelianAnomaly}, so in particular anti-field independent. Then \eqref{eq:cD_A_A_S} can be rewritten as
\beq
\label{eq:hat_cD_A_A_S}
 \hat \cD^\lambda_{\ba} A^m(\et{S_\ia}) - A^m((S, \bar \delta_{\lambda \ba} \Psi) \otimes \et{S_\ia} ) \sim 0,
\eeq
where we used $\hat \cD^\lambda_\ba \defeq \bar \delta_{\lambda \ba} - \delta_\ba - ( \cdot, \cD^\lambda_\ba \Psi )$. 
Considering this at lower order in $\hbar$, this in particular implies
\beq
\label{eq:LowerOrder_1a}
 A^k( (S, \bar \delta_{\lambda \ba} \Psi) \otimes \et{S_\ia}) \sim 0 \qquad \forall k < m.
\eeq
We now use the following corollary of the $\hbar$ expanded consistency condition \eqref{eq:ExpandedConsistencyCondition}:\footnote{For present purposes, only the result \eqref{eq:ExpandedCC} is relevant. Later, also \eqref{eq:ExpandedCC_2} will be used.}
\begin{corollary}
\label{cor:cc}
Let $A^m(\et{S_\ia})$ be the first non-trivial term (\wrt $\sim$) in the $\hbar$ expansion of $A(\et{S_\ia})$. If $A^k(F \otimes \et{S_\ia})$, $k \leq m$ is the first non-trivial term in the $\hbar$ expansion of $A(F \otimes \et{S_\ia})$, then 
\beq
\label{eq:ExpandedCC}
 (S, A^k(F \otimes \et{S_\ia})) + (F, A^k(\et{S_\ia})) - A^k( (S,F) \otimes \et{S_\ia}) \sim 0.
\eeq
Also, if $A^m(F_{1/2} \otimes \et{S_\ia})$ are the first non-trivial terms in the $\hbar$ expansion of $A(F_{1/2} \otimes \et{S_\ia})$, and
$A^k(F_1 \otimes F_2 \otimes \et{S_\ia})$, $k \leq m$, is the first non-trivial term in the $\hbar$ expansion of $A(F_1 \otimes F_2 \otimes \et{S_\ia})$, then
\begin{multline}
\label{eq:ExpandedCC_2}
 (S, A^k(F_1 \otimes F_2 \otimes \et{S_\ia})) + (F_1, A^k(F_2 \otimes \et{S_\ia})) + (F_2, A^k(F_1 \otimes \et{S_\ia})) \\
 - A^k( (S,F_1) \otimes F_2 \otimes \et{S_\ia}) - A^k( F_1 \otimes (S,F_2) \otimes \et{S_\ia}) \sim 0.
\end{multline}
Here $F_1$, $F_2$ are assumed to be bosonic (otherwise the signs need to be adjusted).
\end{corollary}

From \eqref{eq:LowerOrder_1a} and \eqref{eq:ExpandedCC}, we  may thus conclude that $(S, A^k(\bar \delta_{\lambda \ba} \Psi \otimes \et{S_\ia})) \sim 0$ for $k < m$. Hence, the four-form field $\xi^k_\ba \defeq A^k(\bar \delta_{\lambda \ba} \Psi \otimes \et{S_\ia})|_{\lambda = 1}$ (recall the definition of $\cdot |_{\lambda = 1}$ below Lemma~\ref{lemma:field}), which is of vanishing ghost number, mass dimension four (we count $\ba$ as having mass dimension~1), and linear in $\ba$, is $s$ closed modulo $\ud$.
Furthermore, as $\Psi$ is independent of the abelian part of the background connection, we can restrict to perturbations $\ba$ which are purely in the non-abelian directions. To such fields, the following applies:

\begin{lemma}
\label{lemma:xi_ba_cohomology}
Fields $\xi_\ba$ of mass dimension four which are linear in a perturbation $\ba$ purely in the non-abelian directions are trivial in $H^0_4(s | \ud)$.
\end{lemma}

\begin{proof}
A general element $\xi_\ba \in H^0_4(s | \ud)$ is of the form $\xi_\ba = \Tr_N (\ba \wedge * D)$ with $D$ a section of $\p \otimes \Omega^1$ of vanishing ghost number and mass dimension 3 (one mass dimension is contributed by $\ba$) and $*$ the Hodge dual. Furthermore, $\xi_\ba$ being $s$ closed modulo $\ud$ implies that $s D = P_\lin E$, where $P_\lin$ is the linearized Yang-Mills operator (which vanishes on $\ba$, so that $s \xi_\ba$ is a total derivative) for some $E$ of ghost number and mass dimension 1. It must thus be a linear combination of $\bar \nabla C^I$ and $[A, C]^I$. But as $P_\lin \circ \bar \nabla = 0$, the first term can always be added to $E^I$, so unless it is trivial, we can write it as $E^I = c ( \bar \nabla C + [A, C] )^I = c s A^I$ with a numerical coefficient $c$. Again by $P_\lin \ba = 0$, $\Tr_N (\ba \wedge * P_\lin A)$ is a total derivative, so $\xi'_\ba = \Tr_N (\ba \wedge * (D - c P_\lin A))$ is in the same cohomology class as $\xi_\ba$. However, the corresponding $D' = D - c P_\lin A$ is now $s$ closed: $s D' = P_\lin( E - E ) = 0$. We are thus looking for an $s$ closed Lie algebra valued one-form $D$ of vanishing ghost number and mass dimension 3. Hence, $D$ must at the same time transform covariantly under background gauge transformations and be invariant under a gauge transformation of the dynamical fields. It follows from the well-known characterization of $H^0(s)$ that it must either be trivial, for example $D^I = s \bar \nabla \tilde C^I$, or a c-number. By power counting, the only possibility is $D^I_\nu = \bar \nabla^\mu \bar F_{\mu \nu}$, which however vanishes by our assumption that the background is on-shell.
\end{proof}

Analogously to the removal of the gauge anomaly, we can thus, by redefinition of time-ordered products $T( \bar \delta_{\lambda \ba} \Psi \otimes \eti{S_\ia} )$ (which does not spoil perturbative agreement), achieve that 
\beq
\label{eq:LowerOrder_1b}
 A^k(\bar \delta_{\lambda \ba} \Psi \otimes \et{S_\ia}) \sim 0  \qquad \forall k < m.
\eeq
Now as both $\bar \delta_{\lambda \ba} \Psi$ and the anomaly $a^m$ are anti-field independent, it follows from \eqref{eq:LowerOrder_1a}, \eqref{eq:LowerOrder_1b}, and \eqref{eq:ExpandedCC} that 
\beq
\label{eq:s_A_bar_delta_Psi}
 (S, A^m(\bar \delta_{\lambda \ba} \Psi \otimes \et{S_\ia})) \sim A^m((S, \bar \delta_{\lambda \ba} \Psi) \otimes \et{S_\ia}).
\eeq
In particular, with \eqref{eq:hat_cD_A_A_S},
\beq
\label{eq:cD_A_s_A}
 \hat \cD^\lambda_{\ba} A^m(\et{S_\ia}) \sim (S, A^m(\bar \delta_{\lambda \ba} \Psi \otimes \et{S_\ia})).
\eeq
This implies that there must be fields $b^m_\ba$, $c^m_\ba$ such that
\beq
\label{eq:D_a_a}
 \hat \cD_\ba a^m = s b^m_\ba + \ud c^m_\ba.
\eeq
As $\hat \cD_\ba$ commutes with $\ud$ and $s$, \cf \eqref{eq:s_cD}, it is well-defined on the cohomology $H(s | \ud)$. The above shows that if the assumptions of Theorem~\ref{thm:PA} are fulfilled, then $\hat \cD_\ba a^m$ is trivial in $H^1(s | \ud)$, i.e., $a^m$ is background independent in cohomology. If the non-abelian anomaly \eqref{eq:NonabelianAnomaly} turned out to be not background independent in cohomology, this would already prove Theorem~\ref{thm:main}. However, as shown in \cite{ManesStoraZumino}, the non-abelian anomaly \eqref{eq:NonabelianAnomaly} fulfills \eqref{eq:D_a_a}, with $b^m_\ba$ explicitly given by
\beq
\label{eq:b_a_m}
 b_\ba^m = \Tr_N \left[ - \tfrac{1}{2} A (\ba \bar F + \bar F \ba) + \tfrac{1}{6} \ba (A \bar \nabla A + \bar \nabla A A) + \tfrac{1}{12} A^2 (\ba A - A \ba) \right].
\eeq
Hence, \eqref{eq:D_a_a} can not yet be used to rule out the non-abelian anomaly \eqref{eq:NonabelianAnomaly}. Nevertheless, we have learned something: By comparing \eqref{eq:D_a_a} and \eqref{eq:cD_A_s_A} and recalling the definition of $\xi^m_\ba$ given below Corollary~\ref{cor:cc}, we have
\beq
\label{eq:psi_b}
 \xi^m_\ba = b^m_\ba + d^m_\ba
\eeq
for some $d^m_\ba \in H^0(s | \ud)$ (in particular, this includes $\ud$ exact terms). This ambiguity is due to the fact that \eqref{eq:D_a_a} defines $b^m_{\ba}$ only up to terms in the kernel of $s$ (modulo $\ud$). Furthermore, a trivial element of $H^0(s | \ud)$ can be removed from $\xi^m_\ba$ by a redefinition of $T(\bar \delta_{\lambda \ba} \Psi \otimes \eti{S_\ia})$, which can be performed without destroying perturbative agreement. However, by Lemma~\ref{lemma:xi_ba_cohomology}, the relevant cohomology class is trivial, so that we can ignore the ambiguity $d^m_\ba$ in \eqref{eq:psi_b}.

We would now like to compute a further ``background derivative'', namely of \eqref{eq:psi_b}. As already indicated below \eqref{eq:cD_cD}, we promote $\hat \cD$ to a differential on the tangent space $T \sS$ of background variations, and interpret \eqref{eq:psi_b} as an equality of background variation one-forms. To deal with background variation $p$-forms, we introduce some notation. A general background variation $p$-form field $H$ can be expressed as
\beq
 H^{\mu_1 \dots \mu_p}_{I_1 \dots I_p}(x) \hat \cD \Xi^{I_1}_{\mu_1}(x) \wedge \dots \wedge \hat \cD \Xi^{I_p}_{\mu_p}(x)
\eeq
with $I_k$ Lie algebra indices and $\hat \cD \Xi$ dual to a tangent vector $\ba_\mu^I(x)$, meaning that the evaluation of a tangent vector $\ba$ in $\hat \cD \Xi$ is given by $\hat \cD \Xi^I_\mu(x)(\ba) = \ba^I_\mu(x)$. On such an expression, the differential $\hat \cD$ acts as
\beq
 \hat \cD \left( H^{\mu_1 \dots \mu_p}_{I_1 \dots I_p}(x) \hat \cD \Xi^{I_1}_{\mu_1}(x) \wedge \dots \wedge \hat \cD \Xi^{I_p}_{\mu_p}(x) \right) = (-1)^{\eps_H} \hat \cD_{I_0}^{\mu_0} H^{\mu_1 \dots \mu_p}_{I_1 \dots I_p}(x) \hat \cD \Xi^{I_0}_{\mu_0}(x) \wedge \dots \wedge \hat \cD \Xi^{I_p}_{\mu_p}(x),
\eeq
where $\eps_H$ is the Grassmann parity of $H^{\mu_1 \dots \mu_p}_{I_1 \dots I_p}(x)$ and $\hat \cD_{I}^{\mu}$ is defined by
\beq
 \hat \cD_\ba H(x) = \hat \cD_I^\mu H(x) \ba^I_\mu(x).
\eeq
By \eqref{eq:cD_cD}, \eqref{eq:s_cD}, the differential $\hat \cD$ is nilpotent and anticommutes with $s$. On expressions of the form \eqref{eq:AnomalyExpansion}, one similarly defines $\hat \cD^\lambda$ (due to the presence of the cut-off $\lambda$ this is a differential only \wrt the relation $\sim$).

In the notation just introduced, $A^m(\bar \delta_{\lambda \ba} \Psi \otimes \et{S_\ia})$ can be seen as the evaluation of the background variation one form $A^m(\bar \delta^{\lambda \mu}_I \Psi \otimes \et{S_\ia}) \hat \cD \Xi_\mu^I$ in $\ba$ (the superscript $\lambda$ at $\bar \delta$ indicates multiplication with the cut-off $\lambda$). As the gauge fixing fermion is of second order in the dynamical fields and does not contain anti-fields, and the anomaly vanishes if one of the factors is a linear field \cite[Lemma~A.2]{BgInd}, we may replace $\bar \delta^{\lambda \mu}_I \Psi$ by $\hat \cD^{\lambda \mu}_I \Psi$ in this expression. We now compute a $\hat \cD^\lambda$ differential of this background variation one-form. A further application of field independence \eqref{eq:AnomalyFieldIndependence}, background independence \eqref{eq:PA_Thm_2}, and \eqref{eq:BrokenBackgroundIndependenceCutoff} yields, for $k \leq m$,
\begin{multline}
\label{eq:hat_cd_A_delta_Psi}
 \hat \cD^\lambda \left( A^k(\bar \delta^{\lambda \mu}_I \Psi \otimes \et{S_\ia}) \hat \cD \Xi_\mu^I \right) \\
 \sim \left( A^k( \hat \cD^{\lambda \mu_0}_{I_0} \hat \cD^{\lambda \mu_1}_{I_1} \Psi \otimes \et{S_\ia} ) + A^k( \hat \cD^{\lambda \mu_1}_{I_1} \Psi \otimes (S, \hat \cD^{\lambda \mu_0}_{I_0} \Psi) \otimes \et{S_\ia} ) \right) \hat \cD \Xi_{\mu_0}^{I_0} \wedge \hat \cD \Xi_{\mu_1}^{I_1}.
\end{multline}
The first term on the \rhs vanishes (\wrt $\sim$), by nilpotency of $\hat \cD^\lambda$. To deal with the second term, we note that, analogously to \eqref{eq:LowerOrder_1a}, \eqref{eq:LowerOrder_1b}, we have, for all $k < m$,\footnote{The second identity follows directly from \eqref{eq:hat_cd_A_delta_Psi} and \eqref{eq:LowerOrder_1b}. From the second identity, \eqref{eq:LowerOrder_1b}, and \eqref{eq:ExpandedCC_2}, we conclude that the field corresponding to the \lhs of the first identity is $s$ closed (modulo $\ud$). But there are no such fields of ghost number $-1$, anti-symmetric and linear in $\ba$, $\ba'$ and of mass dimension four (two of which are already contributed by $\ba$ and $\ba'$): By power counting, the only possible fields with these properties are $\Tr_N ([\ba^\mu, \ba'_\mu] B^{\ddag})$ and $\Tr_N ([\ba^\mu, \ba'_\mu] \tilde C ) \vol$, but no linear combination of these is $s$ closed modulo $\ud$.}
\begin{align}
\label{eq:LowerOrder_2}
 A^k(\bar \delta_{\lambda \ba} \Psi \otimes \bar \delta_{\lambda \ba'} \Psi \otimes \et{S_\ia}) & \sim 0, &
 A^k( (S, \bar \delta_{\lambda \ba} \Psi) \otimes \bar \delta_{\lambda \ba'} \Psi \otimes \et{S_\ia}) & \sim 0.
\end{align}
Then, with \eqref{eq:ExpandedCC_2}, we have (recall that $\Psi$ is fermionic)
\begin{align}
 & (S, A^m( \hat \cD^{\lambda \mu_1}_{I_1} \Psi \otimes \hat \cD^{\lambda \mu_0}_{I_0} \Psi \otimes \et{S_\ia} )) \nn \\
 & \sim A^m( (S, \hat \cD^{\lambda \mu_1}_{I_1} \Psi) \otimes \hat \cD^{\lambda \mu_0}_{I_0} \Psi \otimes \et{S_\ia} ) - A^m( \hat \cD^{\lambda \mu_1}_{I_1} \Psi \otimes (S, \hat \cD^{\lambda \mu_0}_{I_0} \Psi) \otimes \et{S_\ia} ) \nn \\
 & \sim A^m( \hat \cD^{\lambda \mu_0}_{I_0} \Psi \otimes (S, \hat \cD^{\lambda \mu_1}_{I_1} \Psi) \otimes \et{S_\ia} ) - A^m( \hat \cD^{\lambda \mu_1}_{I_1} \Psi \otimes (S, \hat \cD^{\lambda \mu_0}_{I_0} \Psi) \otimes \et{S_\ia} ).
\end{align}
Hence,
\beq
 \hat \cD^\lambda \left( A^m(\bar \delta^{\lambda \mu}_I \Psi \otimes \et{S_\ia}) \hat \cD \Xi_\mu^I \right) \sim - \tfrac{1}{2} (S, A^m( \hat \cD^{\lambda \mu_1}_{I_1} \Psi \otimes \hat \cD^{\lambda \mu_0}_{I_0} \Psi \otimes \et{S_\ia} ) ) \hat \cD \Xi_{\mu_0}^{I_0} \wedge \hat \cD \Xi_{\mu_1}^{I_1}.
\eeq
In terms of fields, this means that, modulo $\ud$ exact terms,
\beq
\label{eq:hat_d_b_s_chi}
 \hat \cD b^m = \hat \cD \xi^m = s \chi^m,
\eeq
with $\chi^m = \frac{1}{2} A^m( \hat \cD^\lambda \Psi \otimes \hat \cD^\lambda \Psi \otimes \et{S_\ia} ) |_{\lambda = 1}$.

On the other hand, we can directly compute the \lhs of \eqref{eq:hat_d_b_s_chi} for $b^m$ corresponding to the non-abelian anomaly, i.e., given by \eqref{eq:b_a_m}. We obtain, up to $\ud$ exact terms,
\beq
\label{eq:cD_chi}
 ( \hat \cD b^m )(\ba_1, \ba_2) = \tfrac{1}{2} \Tr_N \left[ (\ba_1 \ba_2 - \ba_2 \ba_1) \bar F \right].
\eeq
The r.h.s., being a c-number and not $\ud$ exact, is not $s$ exact modulo $\ud$. 
It vanishes for general background connections and variations $\bar a_1$, $\bar a_2$ thereof if and only if the $d$ symbol of the trace $\Tr_N$ vanishes (in which case also the non-abelian anomaly \eqref{eq:NonabelianAnomaly} vanishes, as discussed below \eqref{eq:AbelianAnomaly}).
Hence, \eqref{eq:hat_d_b_s_chi} is not fulfilled. Thus, also the non-abelian anomaly can not be present under the assumptions of Theorem~\ref{thm:main}. This concludes the proof of Theorem~\ref{thm:main}.



\begin{remark}
The expression on the \rhs of \eqref{eq:cD_chi} also occurs as the obstruction $E(a_1, a_2)$, \cf \eqref{eq:E}, to remove a violation of perturbative agreement indicated by a non-vanishing covariant divergence of the free current of the form $T( \bar \nabla_\mu J^{\mu I}_0 ) \propto \Tr_N [ T^I \bar F \bar F ]$, as shown in \cite{GlobalAnomalies}.
This indicates a connection between the cohomological classification of obstructions to perturbative agreement and anomaly freedom (at least for the non-abelian anomalies).
\end{remark}

\section{Conclusion}
\label{sec:Conclusion}

We have proven under quite general conditions that the absence of gauge anomalies can be determined at the one-loop level. As our proof is based on perturbative agreement, which is a formalization of background independence (but see \cite{BgInd} for a thorough discussion of background independence in gauge theories), one might also phrase our result as stating that background independence \wrt changes in the background connection implies absence of gauge anomalies.

A shortcoming of our analysis is that we restricted to vanishing background bosonic matter fields. We briefly sketch how to include these and how to possibly extend our result to that case. At the level of the action, non-trivial background matter fields can be simply introduced by replacing the bosonic matter field $\vp$ by $\bar \Phi + \vp$, with $\bar \Phi$ the matter background field. The action of the BV operator on $\vp$ then has to be modified to $s \vp = - \rho(C) (\bar \Phi + \vp)$ so that also the free part $s_0$ acts non-trivially on the bosonic matter fields. In the gauge fixing fermion \eqref{eq:Xi}, also a term $\int \tilde C^I H^I(\bar \Phi, \vp) \vol$ would be included, allowing for an $R_\xi$ gauge (in which terms with single derivatives are eliminated from the free part of the action). 

As for anomalies, the non-abelian anomaly \eqref{eq:NonabelianAnomaly} is independent of matter fields (in particular background matter fields), so if it is absent for vanishing matter background fields (as is guaranteed if the assumptions of Theorem~\ref{thm:main} are fulfilled), it will also be absent for non-trivial matter background fields. The same is true for the abelian anomalies of the form \eqref{eq:AbelianChiral}, which are well-known to potentially occur. It remains to rule out abelian anomalies depending on the background matter field. As discussed in Remark~\ref{rem:AbelianAnomaly}, one may hope to rule these out on general grounds. Alternatively, one can slightly strengthen the requirements of Theorem~\ref{thm:main} such that also perturbative agreement \wrt changes in the background matter field holds and rule out the occurrence of anomalies depending explicitly on the background matter field $\bar \Phi$ (not just on the combination $\bar \Phi + \vp$ which has already been ruled out by Theorem~\ref{thm:main} for $\bar \Phi = 0$) by using again the background independence of the anomaly.

Finally, we remark that the technique used here (which was essentially model-independent) can apparently not be straightforwardly used to prove the original version of the Adler-Bardeen theorem \cite{AdlerBardeen}, i.e., the non-renormalization of the chiral anomaly in QED: An essential ingredient of our proof is that a simple $\cO(\hbar)$ criterion decides on whether perturbative agreement can be fulfilled. To show this, one proceeds inductively (in $\deg_\Phi$, the total number of fields) \cite{HollandsWaldStress, BackgroundIndependence}. This procedure is model independent in the sense that the specific model considered only enters through the $\cO(\hbar)$ criterion that needs to be checked (the on-shell vanishing of the free current). However, if the relevant criterion is not fulfilled at lowest order, then not only will perturbative agreement not be fulfilled at higher order, but we also lose control about how it is violated, at least unless one engages in a model-dependent analysis.

\subsection*{Acknowledgements}

I would like to thank Stefan Hollands for useful discussions.


\end{document}